\documentclass{article}
\usepackage{graphicx} 
\usepackage{bbm}
\usepackage{amssymb}
\usepackage{float}
\usepackage{amsmath}
\usepackage{hyperref}
\usepackage{amsthm}
\newtheorem{lemma}{Lemma}
\usepackage{dsfont}
\usepackage{authblk}
\usepackage{marvosym}
\newtheorem{theorem}{Theorem}
\newtheorem{corollary}{Corollary}
\usepackage[sorting=none,style=nature]{biblatex} 
\date{}

\title{Sample complexity bounds for the Jensen-Shannon divergence}
\author[1]{\mbox{Oren Richter}}
\author[1]{\mbox{Adi Ben-Ari}}
\author[1]{\mbox{Tom Talpir}}
\author[1,\Letter]{\mbox{Elad Schneidman}}

\affil[1]{Department of Brain Sciences, Weizmann Institute of Science, Rehovot 76100, Israel}
\affil[\Letter]{elad.schneidman@weizmann.ac.il}
\begin{document}
\maketitle

\begin{abstract}
The Jensen-Shannon divergence (JSD) is a symmetric and bounded measure of the dissimilarity of two probability distributions, which has become a standard tool in statistics, information theory, and machine learning. We complement the understanding of its mathematical properties by presenting an analysis of the amount of data that is needed to distinguish between two distributions, given the value of JSD between them. We find the number of independent and identically distributed samples that suffice for a classifier to determine which of two distributions generated observed data at a desired error rate, for two complementary classifiers: we show that for the log-likelihood-ratio classifier, a sample size that grows as the inverse JSD is sufficient, whereas for a majority-vote classifier assembled from independent single-sample decisions, the sufficient size grows as the squared inverse JSD. These distinct scalings offer operational readings of JSD values and their translation into distinguishability in different contexts. 
\end{abstract}
 
\section*{Introduction}
The Kullback-Leibler (KL) divergence between two probability distributions $P(x)$ and $Q(x)$ , 
\begin{equation} 
D_{KL}[P(x)||Q(x)] = \sum_x P(x) \log_2 \frac{P(x)}{Q(x)} 
\end{equation} 
has been a prominent measure of the dissimilarity of probability distributions \cite{kl_divergence_1951} due to its foundational role in information theory, and its coding-related interpretations as a measure of coding inefficiency or distinguishability of sources \cite{cover_elements_2005}. Notably, the utility of KL and its interpretation are limited by its asymmetric nature, and because it diverges if there is an $x$ for which $P(x)>0$ whereas $Q(x)=0$. 

The Jensen-Shannon divergence (JSD) is a symmetric and finite measure of the dissimilarity of probability distributions \cite{lin_divergence_1991}, which has become a popular tool in many data-oriented applications \cite{grosse_analysis,sims_alignment-free,goodfellow_generative_2014}, as well as an interesting measure from a theoretical perspective. JSD measures the dissimilarity of $P(x)$ and $Q(x)$, by weighting their respective KL dissimilarity to an intermediate distribution $M(x)$, namely, 

\begin{equation} D_{JS}[P(x)||Q(x)] = \lambda D_{KL}[P(x)||M(x)] +(1-\lambda) D_{KL}[Q(x)||M(x)] 
\end{equation} 
where $M(x)=\lambda P(x)+ (1-\lambda) Q(x)$, and  $\lambda$ is a fraction between 0 and 1 (commonly taken to be equal to $\frac{1}{2}$, a convention we also use here). Importantly, the value of $D_{JS}$ is bounded, ranging from 0 for identical distributions, to 1 bit for non-overlapping distributions with disjoint supports. Moreover, it belongs to the family of $f$-divergences, inheriting their information-monotonicity under coarse graining \cite{csiszar_information_67, csiszar_information_2004,amari_alpha}; and
$\sqrt{D_{JS}[P(x)||Q(x)]}$ is a proper metric, satisfying the triangle inequality \cite{endres_new_metric}. These properties have made the JSD popular across different fields, from statistics and information theory to data science and machine learning, where it appears in two-sample testing, generative modeling, and representation learning, among many other settings.
 
Despite its widespread use, the ``operational'' meaning of a given JSD value is often left implicit. While there is a known bound on the Bayes classification error in distinguishing between two probability distributions from a single observation based on knowing that $D_{JS}[P(x)||Q(x)]=d$ \cite{lin_divergence_1991}, it is not immediately clear how JSD governs distinguishability from many samples.

We therefore ask here how many independent and identically distributed (i.i.d.) samples are sufficient for a classifier to identify the source distribution at a desired classification error rate. The sample complexity of distinguishing two distributions from i.i.d. samples is a classical quantity in statistical decision theory, and under uniform prior over the two distributions, as we consider here, it has been shown to be given by the Hellinger distance between the two distributions \cite{bar_yossef_thesis, pensia_sample}, up to multiplicative constants. We present here a complementary, direct, and self-contained derivation, tracking explicit constants in the relation between JSD and sample size for the log-likelihood-ratio classifier. This result gives us a baseline for the second bound we derive, characterizing the sample size for the majority-vote classifier built from single-sample decisions. The two analyses yield qualitatively different scalings ($1/d$ versus $1/d^{2}$), reflecting the interplay between JSD values and classification power, and expanding our understanding and intuition of JSD.

\section*{Results }

We consider the problem of binary hypothesis testing (also known as binary detection), given two probability distributions with support $\mathcal{X}$, which we denote $P$ and $Q$: observing an i.i.d.\ sample $X_{1:N}=\{x_1,\dots,x_N\} \in \mathcal{X}^N$ generated under one of two competing hypotheses,
\begin{equation}
H_0: X_{1:N} \overset{\text{iid}}{\sim} P, \qquad H_1: X_{1:N} \overset{\text{iid}}{\sim} Q
\end{equation}
our goal is to decide, from the sample $X_{1:N}$, whether $P$ or $Q$ is the generating distribution. A decision rule (or a classifier) is a map $\mathcal{C}: \mathcal{X}^N \to \{0,1\}$, whose performance is characterized by the Type-I and Type-II error probabilities \cite{neyman_ix_1933,cover_elements_2005}, explicitly given by:
\begin{equation}
    \alpha(\mathcal{C}) = \sum_{ X_{1:N} \in R_Q} \prod_{i=1}^N P(x_i) \ , \ 
    \beta(\mathcal{C}) = \sum_{ X_{1:N} \in R_P} \prod_{i=1}^N Q(x_i)
\end{equation}

\noindent where 
\begin{equation}
    R_Q = \{ X_{1:N} \in \mathcal{X}^N \ |  \ \mathcal{C}(X_{1:N}) = 1 \} \ , \   R_P = \{ X_{1:N} \in \mathcal{X}^N \ |  \ \mathcal{C}(X_{1:N}) = 0 \}
\end{equation}
\noindent are the rejection regions of $H_0$ and $H_1$, respectively.

\subsection*{Sample complexity of optimal binary detection scales as the inverse JSD}

Under uniform prior over the two hypotheses, the rule that minimizes the Bayes probability of error $p_e^{(N)}=\tfrac{1}{2} \alpha(\mathcal{C}) + \tfrac{1}{2} \beta(\mathcal{C})$ is the likelihood-ratio test
\begin{equation}\label{LLR_classifier}
\mathcal{C}_{\text{LLR}}(X_{1:N}) = \mathbbm{1} \left[ \sum_{i=1}^N \log \frac{Q(x_i)}{P(x_i)} > 0 \right],
\end{equation}

\noindent where $\mathbbm{1}[\cdot]$ is the indicator function.
We recall that the Neyman-Pearson lemma \cite{neyman_ix_1933} states that, among all tests with the same Type-I error rate, this test is the most powerful.

We next quantify the sample complexity of the LLR classifier in terms of the Jensen-Shannon divergence.

\begin{theorem}[Sample complexity of the LLR classifier]
\label{thm:sample_complexity}
Let $P$ and $Q$ be two probability distributions with common support $\mathcal{X}$, and let $d = D_{JS}(P,Q)$ denote their Jensen-Shannon divergence. For a target error rate $\epsilon > 0$, the Bayes probability of error of the LLR classifier $\mathcal{C}_{\text{LLR}}$ on an i.i.d.\ sample of size $N$ satisfies $p_e^{(N)} \leq \epsilon$ when
\begin{equation}
    N \geq \frac{\log(1/\epsilon)}{d \, \log(2)}.
\label{eq:thm_N_bound}
\end{equation}
\end{theorem}

The proof relies on bounding the Bayes error of $\mathcal{C}_{\text{LLR}}$ via the Chernoff information, then lower bounding the Chernoff information by the Jensen-Shannon divergence through a chain of three lemmas, and combining these in a corollary that completes the proof.

The Bayes probability of error $p_e^{(N)}$ of $\mathcal{C}_{\text{LLR}}$ is upper bounded by 
\begin{equation}
    p_e^{(N)} \leq e^{-N\,C(P,Q)},
\label{eq:chernoff_bound}
\end{equation}
where $C(P,Q)$ is the Chernoff information \cite{chernoff_measure_1952}, defined as
\begin{equation}
    C(P,Q) = \max_{\gamma\in[0,1]} -\log \sum_{x\in\mathcal{X}} P(x)^{\gamma}\,Q(x)^{1-\gamma},
\label{eq:chernoff_def}
\end{equation}
with $\mathcal{X}$ the common support of $P$ and $Q$ and $\log$ the natural logarithm (chapter 11.9 in \cite{cover_elements_2005}).

To relate \eqref{eq:chernoff_bound} to the Jensen-Shannon divergence, we establish the lower bound
\begin{equation}
    C(P,Q) \geq \log(2)\cdot D_{JS}(P,Q),
\label{eq:chernoff_js_bound}
\end{equation}
through three lemmas, which we then chain together. The argument relies on the following quantities. The Bhattacharyya coefficient is
\begin{equation}
    \mathrm{BC}(P,Q) = \sum_{x\in\mathcal{X}} \sqrt{P(x)\,Q(x)},
\label{eq:bhattacharyya_coeff}
\end{equation}
the Bhattacharyya distance is
\begin{equation}
    B(P,Q) = -\log \mathrm{BC}(P,Q),
\label{eq:bhattacharyya_def}
\end{equation}
and the squared Hellinger distance is
\begin{equation}
    H^{2}(P,Q) = \tfrac{1}{2}\sum_{x\in\mathcal{X}} \big(\sqrt{P(x)} - \sqrt{Q(x)}\big)^{2}
    = 1 - \mathrm{BC}(P,Q).
\label{eq:hellinger_def}
\end{equation}
 
\begin{lemma}[The Chernoff information dominates the Bhattacharyya distance]
\label{lem:chernoff_bhatt}
$C(P,Q) \geq B(P,Q)$.
\end{lemma}
\begin{proof}
Denoting \[f(\gamma):=-\log\sum_x P(x)^{\gamma}Q(x)^{1-\gamma},\] it immediately follows from definitions \eqref{eq:chernoff_def},\eqref{eq:bhattacharyya_def} that
\begin{equation}
    C(P,Q) = \max_{\gamma\in[0,1]}f(\gamma) \geq f(\tfrac{1}{2}) = B(P,Q).
    \label{eq:step1}
\end{equation}
\end{proof}
 
\begin{lemma}[The Bhattacharyya distance dominates the squared Hellinger distance]
\label{lem:bhatt_hellinger}
$B(P,Q) \geq H^{2}(P,Q)$.
\end{lemma}
\begin{proof}
Since the function $\log(y)$ is concave, it lies below its tangent line at $y=1$, which is $y-1$, namely $\log y \leq y -1 \quad \forall y>0$. Thus $-\log(y) \geq 1 - y$. Substituting $y = \mathrm{BC}(P,Q)$ (which is positive since $P,Q$ are probability distributions, namely non-negative and with sum 1),
\begin{equation}
    B(P,Q) = -\log \mathrm{BC}(P,Q) \geq 1 - \mathrm{BC}(P,Q) = H^{2}(P,Q),
\label{eq:step2}
\end{equation}
where the last equality is \eqref{eq:hellinger_def}.
\end{proof}
 
\begin{lemma}[Squared Hellinger distance dominates the Jensen-Shannon divergence]
\label{lem:hellinger_djs}
$H^{2}(P,Q) \geq \log(2)\cdot D_{JS}(P,Q)$.
\end{lemma}
\begin{proof}
We claim the inequality holds term-wise. Writing $0 \leq a = P(x)\leq 1$ and $0 \leq b = Q(x)\leq 1$, define
\begin{equation}
    \varphi(a,b) = \tfrac{1}{2}a\log\frac{2a}{a+b} + \tfrac{1}{2}b\log\frac{2b}{a+b}, \qquad
    \psi(a,b) = \tfrac{1}{2}(a+b) - \sqrt{ab},
\label{eq:phi_psi_def}
\end{equation}
with the convention $0\log 0 := 0$, so that $D_{JS}(P,Q) = \frac{1}{\log(2)}\sum_x \varphi(P(x),Q(x))$ and $H^{2}(P,Q) = \sum_x \psi(P(x),Q(x))$. It therefore suffices to prove $\varphi(a,b)\leq\psi(a,b)$ for all $0 \leq a,b\leq 1$.
 
The case where $a=b=0$ is trivial, since $\varphi(a,b) = \psi(a,b) = 0$. Otherwise, $a+b > 0$. We observe that both $\varphi$ and $\psi$ are positively homogeneous of degree one, namely $\varphi(\lambda a,\lambda b)=\lambda\varphi(a,b)$ and $\psi(\lambda a,\lambda b)=\lambda\psi(a,b)$ for every $\lambda>0$. Thus, we can assume with no loss of generality that $a+b=1$ (because homogeneity implies that $\varphi(\frac{a}{a+b},\frac{b}{a+b})\leq\psi(\frac{a}{a+b},\frac{b}{a+b}) \iff \varphi(a,b)\leq\psi(a,b)$). This observation reduces the claim to the one-variable inequality
\begin{equation}
    F(u) \leq G(u), \qquad u\in[0,1],
\label{eq:one_var_ineq}
\end{equation}
where
\begin{equation}
    F(u) := \varphi(u,1-u)  = \tfrac{1}{2}\big[u\log(2u) + (1-u)\log(2(1-u))\big]
\label{eq:F_def}
\end{equation}
and
\begin{equation}
    G(u) := \psi(u, 1-u) = \tfrac{1}{2} - \sqrt{u(1-u)}.
\label{eq:G_def}
\end{equation}
 
To prove \eqref{eq:one_var_ineq}, define $A(u) = G(u) - F(u)$. We show that $A(u)\geq 0$ for all $u\in[0,1]$.
First, we observe that $A(u)=A(1-u)$, namely $A(u)$ is symmetric around $\tfrac{1}{2}$ in $[0,1]$. Additionally, for $u=\tfrac{1}{2}$ we get $F(u)=G(u)=0$ and thus $A(u)=0$. Thus, to show that $A(u)\geq 0$ it is sufficient to show that $A'(u)\geq0 \quad \forall u\in[\tfrac{1}{2},1)$ (because then $A$ is non-decreasing on $[\tfrac{1}{2},1)$, so that $A(u)\geq A(\tfrac{1}{2})=0$). We establish this by reusing the same argument for $A'$: we show that $A'(\tfrac{1}{2})=0$ and that $A''(u) \geq 0 \quad \forall u\in(0,1)$; the latter implies $A'$ is non-decreasing on $[\tfrac{1}{2},1)$, so that $A'(u)\geq A'(\tfrac{1}{2})=0$ throughout that interval.
 
Taking the first derivative we get
\[A'(u)=\frac{2u-1}{2\sqrt{u(1-u)}}-\frac{1}{2}\log \Big(\frac{u}{1-u} \Big).\]
First, we validate that $A'(\tfrac{1}{2})=0$, which is indeed the case. Next, we rewrite the derivative as
\[A'(u)=\frac{1}{2}\Bigg(\sqrt{\frac{u}{1-u}} - \sqrt{\frac{1-u}{u}} \Bigg) - \log \Bigg(\sqrt{\frac{u}{1-u}} \Bigg).\]
Now, defining $t(u) := \sqrt{\frac{u}{1-u}}$ and $g(t):= \frac{1}{2} \Big(t-\frac{1}{t} \Big) - \log(t)$, the derivative takes the form $A'(u)=g(t(u))$.
Differentiating again via the chain rule we get
\[A''(u)=\frac{d}{dt}g(t)\cdot \frac{d}{du}t(u).\]
To prove that $A''(u)\geq0 \quad \forall u \in (0,1)$ we can show that each component is non-negative separately:
\[\frac{d}{dt}g(t)=\frac{1}{2}\Big( 1 + \frac{1}{t^2} \Big) - \frac{1}{t}=\frac{1}{2}\Big( \frac{1}{t} - 1 \Big)^2 \geq 0 \quad \forall t > 0,\]
and
\[\frac{d}{du}t(u) = \frac{1}{2\sqrt{u}(1-u)^{\tfrac{3}{2}}} \geq 0 \quad \forall u\in(0,1).\]
Hence $A''(u)=\frac{d}{dt}g(t)\cdot \frac{d}{du}t(u) \geq0 \quad \forall u\in (0,1)$, which completes the argument that $A(u)\geq0$ on $[0,1]$.
 
This establishes \eqref{eq:one_var_ineq}, hence $\varphi(a,b)\leq\psi(a,b)$ for all $0\leq a,b\leq 1$, and summing over $\mathcal{X}$ yields
\begin{equation}
    \log(2) \cdot D_{JS}(P,Q) = \sum_{x\in\mathcal{X}} \varphi(P(x),Q(x)) \leq \sum_{x\in\mathcal{X}} \psi(P(x),Q(x)) = H^{2}(P,Q).
\label{eq:step3}
\end{equation}
\end{proof}

\begin{corollary}[Combining the lemmas]
\label{cor:combining}
The Chernoff information and the Jensen-Shannon divergence satisfy $C(P,Q) \geq \log(2)\cdot D_{JS}(P,Q)$, and consequently Theorem~\ref{thm:sample_complexity} holds.
\end{corollary}
\begin{proof}
Chaining Lemmas \ref{lem:chernoff_bhatt}, \ref{lem:bhatt_hellinger}, and \ref{lem:hellinger_djs} gives the desired relation between the Chernoff information and the Jensen-Shannon divergence,
\begin{equation}
    C(P,Q) \;\geq\; B(P,Q) \;\geq\; H^{2}(P,Q) \;\geq\; \log(2) \cdot D_{JS}(P,Q).
\label{eq:chernoff_djs}
\end{equation}
 
Substituting \eqref{eq:chernoff_djs} into the error bound \eqref{eq:chernoff_bound}, and denoting $D_{JS}(P,Q)=d$, the Bayes error rate of $\mathcal{C}_{\text{LLR}}$ satisfies
\begin{equation}
    p_e^{(N)} \leq e^{-N\,C(P,Q)} \leq e^{-N d \log(2)}.
\label{eq:err_djs_bound}
\end{equation}
To guarantee an upper bound of $\epsilon>0$ on the classification error rate of $\mathcal{C}_{\text{LLR}}$, it is then sufficient to require
\begin{equation}
    e^{-N d \log(2)} \leq \epsilon \iff N \geq \frac{\log(1/\epsilon)}{d \log(2)}.
\label{eq:final_N_bound_opt}
\end{equation}
This means that given a desired fixed error rate for the optimal classifier $\mathcal{C}_{\text{LLR}}$ that distinguishes between the distributions $P$ and $Q$ based on an i.i.d. sample, the sample size that is sufficient to meet the error rate is proportional to $\frac{1}{D_{JS}(P,Q)}$.
\end{proof}

\subsection*{Sample complexity of distributed binary detection scales as the squared inverse JSD}

We also consider the scenario of distributed classification, in which the decision is not based directly on the $N$ samples, but is instead aggregated from $N$ independent single-sample decisions. Let $\mathcal{C}_{\text{single}}: \mathcal{X}\to\{0,1\}$ be a fixed single-sample classifier, and let $X_{1:N}=(x_1,\dots,x_N)$ be $N$ i.i.d.\ samples drawn from the true source which is either $P$ or $Q$, with $N$ odd. The majority-vote classifier $\mathcal{C}_{\text{multi}}$ applies $\mathcal{C}_{\text{single}}$ to each sample and outputs the majority label,
\begin{equation}
\mathcal{C}_{\text{multi}}(X_{1:N}) := \mathrm{mode}\big\{\mathcal{C}_{\text{single}}(x_1),\dots,\mathcal{C}_{\text{single}}(x_N)\big\}.
\end{equation}

We note that $\mathcal{C}_{\text{multi}}$ is generally suboptimal relative to the Bayes/Neyman--Pearson-optimal test on the raw samples $X_{1:N}$, since hard-quantizing each sample to a single bit discards the magnitude of its evidence. It is, however, the Bayes-optimal rule for the distributed classification setting, in which each of the $N$ classifiers observes a single sample and must commit to its own decision before the resulting $N$ i.i.d.\ binary decisions are combined. When all $N$ local classifiers are identical, the optimal data-fusion rule of \cite{chair_optimal_1986} reduces exactly to an (unweighted) majority vote.

We now state the main result of this section, which quantifies the sample complexity of the majority-vote classifier in terms of the Jensen-Shannon divergence.

\begin{theorem}[Sample complexity of the majority-vote classifier]
\label{thm:sample_complexity_multi}
Let $P$ and $Q$ be two probability distributions with common support $\mathcal{X}$, and let $d = D_{JS}(P,Q)$ denote their Jensen-Shannon divergence. Let $\mathcal{C}_{\text{multi}}$ be the majority-vote classifier built from $N$ Bayes-optimal single-sample classifiers, with Bayes error $p_e$ each, applied to $N$ i.i.d.\ samples with $N$ odd. For any target error rate $\epsilon > 0$, the error rate of $\mathcal{C}_{\text{multi}}$ is at most $\epsilon$ whenever
\begin{equation}
    N \geq \frac{2\log(1/\epsilon)}{d^2}.
\label{eq:thm_N_bound_multi}
\end{equation}
\end{theorem}

\begin{proof}
Let $E_i\in\{0,1\}$ be a random variable that takes 0 if the $i$-th sample was correctly classified by $\mathcal{C}_{\text{single}}$ and 1 if there was a classification error. The classifier $\mathcal{C}_{\text{multi}}$ errs if and only if the majority of individual classifications were erroneous, namely
\begin{equation}
    \sum_{i=1}^{N}E_i\geq\frac{N}{2}.
\label{eq:maj_vote_err}
\end{equation}
Denoting the Bayes probability error of $\mathcal{C}_{\text{single}}$ by $p_e$, we get that the expected value of $E_i$ is $\mathbb{E}[E_i]=p_e$, and thus we can rewrite \eqref{eq:maj_vote_err} as
\begin{equation}
    \frac{1}{N}\sum_{i=1}^{N}E_i - \mathbb{E}[E_i]\geq \frac{1}{2}-p_e.
\label{eq:hoeff_maj_vote_error}
\end{equation}
Hoeffding (Theorem 1 in \cite{hoeffding_probability_1963}) then gives an upper bound on the probability that the event described by \eqref{eq:hoeff_maj_vote_error} occurs, which is
\begin{equation}
   \Pr\Big[\frac{1}{N}\sum_{i=1}^{N}E_i - \mathbb{E}[E_i]\geq \frac{1}{2}-p_e \Big] \leq \Big(2\sqrt{p_e(1-p_e)}\Big)^{N}.
\label{eq:hoeff_bound}
\end{equation}

Theorem 4 in \cite{lin_divergence_1991} states an upper bound for $p_e$ in terms of $D_{JS}(P,Q)$. Assuming a uniform prior over the source distribution and denoting $D_{JS}(P,Q)=d$, the bound is given by
\begin{equation}
    p_e\leq\frac{1}{2}(1 - d).
\label{eq:lin_bound}
\end{equation}

We note that $0\leq p_e\leq \frac{1}{2}$ and thus using the fact that the function $f(x)=x(1-x)$ is monotonically increasing in $[0,\frac{1}{2}]$, we obtain
\begin{equation}
    p_e(1-p_e) \leq \frac{1}{2}(1-d) \Big(1 - \frac{1}{2}(1-d) \Big) = \frac{1}{4}(1-d^{2}).
\label{eq:single_err_var_bound}
\end{equation}
Substituting \eqref{eq:single_err_var_bound} into \eqref{eq:hoeff_bound}, we bound the error rate of $\mathcal{C}_{\text{multi}}$ further with
\begin{equation}
    \Bigg(2\sqrt{\frac{1}{4}(1-d^{2})}\Bigg)^{N} = (1-d^2)^\frac{N}{2}.
\label{eq:multi_err_djs_bound}
\end{equation}

To guarantee an upper bound of $\epsilon>0$ on the classification error rate of $\mathcal{C}_{\text{multi}}$, we then require
\begin{equation}
    (1-d^2)^\frac{N}{2}\leq \epsilon \iff N \geq \frac{\log(1/\epsilon)}{-\log \Big(\sqrt{1-d^2}\Big)}.
\label{eq:N_bound}
\end{equation}
Reusing the inequality $-\log(y)\geq 1-y$ from Lemma \ref{lem:bhatt_hellinger} and setting $y=1-d^2$, we now bound the denominator from below:
\begin{equation}
    -\log\Big(\sqrt{1-d^2}\Big) = -\tfrac{1}{2}\log\big(1-d^2\big) \geq \tfrac{1}{2}d^2.
\label{eq:log_lower_bound}
\end{equation}
Consequently, any sample size satisfying
\begin{equation}
    N \geq \frac{2\log(1/\epsilon)}{d^2}
\label{eq:final_N_bound}
\end{equation}
also satisfies \eqref{eq:N_bound}, and therefore guarantees an error rate of at most $\epsilon$. This means that given a desired fixed error rate for the classifier $\mathcal{C}_{\text{multi}}$ that distinguishes between the distributions $P$ and $Q$ based on an i.i.d. sample, the sample size that is sufficient to meet the error rate is proportional to $\frac{1}{D_{JS}(P,Q)^2}$. Moreover, expanding the denominator of \eqref{eq:N_bound} around 0 we get that
\begin{equation}
    -\log\sqrt{1-d^2} = \tfrac{1}{2}d^2 + O(d^4),
\end{equation}
thus the bound \eqref{eq:final_N_bound} is tight to leading order as $d\to 0$.
\end{proof}

\section*{Discussion}

We present two operational readings of the Jensen-Shannon divergence in terms of sample complexity. The contrast between the two scalings, $1/d$ for the log-likelihood-ratio classifier and $1/d^{2}$ for the majority-vote classifier, implies that the same JSD value translates into very different data requirements depending on how the evidence in each sample is used. The optimal classifier accumulates the full magnitude of the log-likelihood ratio of every observation, whereas the majority-vote rule first hard-quantizes each sample to a single bit and only then aggregates. Discarding the strength of the per-sample evidence is precisely what costs the extra factor of $1/d$.

These two classifiers can be viewed as bracketing a broader spectrum. The LLR test is Bayes-optimal and therefore sets the best achievable scaling for any classification procedure, $1/d$; the majority vote represents the opposite extreme of maximally coarse local decisions combined by a simple fusion rule (which is optimal with no further assumptions on the individual classifiers). Intermediate strategies, such as soft-quantizing each sample, transmitting a few bits of confidence per observation, or weighting local votes by their confidence, would be expected to interpolate between these regimes. The relevant question in any applied setting is then how much per-sample information one can afford to retain before fusion. When samples must be compressed, communicated under a bit budget, or committed to independently, as in distributed sensing, federated estimation, or biological signaling, the $1/d^{2}$ penalty is the price of locality, and the gap to the $1/d$ optimum quantifies what is lost.

\begin{acknowledgements}
\noindent We thank members of Schneidman’s lab for critical suggestions and insights. This work was supported by Simons Collaboration on the Global Brain grant 542997, Israel Science Foundation grant 137628, Azrieli Institute for Brain and Neural Sciences and the Hedda, Alberto, and David Milman Baron Center for Research on the Development of Neural Networks of the Weizmann institute, as well as the Knell family Institute for Artificial Intelligence, Martin Kushner Schnur, and Mr. \& Mrs. Lawrence Feis. ES is the incumbent of the Joseph and Bessie Feinberg Chair.
\end{acknowledgements}

\printbibliography
\end{document}